\documentclass[a4paper]{article}

\usepackage{amsfonts, amssymb}

\newtheorem{theo}{Theorem}
\newtheorem{defi}{Definition}
\newtheorem{coro}{Corollary}
\newtheorem{prop}{Proposition}
\newtheorem{lemm}{Lemma}

\newenvironment{rema}{\noindent {\bf Remark:}\ }{\ \rule{1mm}{2mm}}
\newcommand{\pr}{\indent{\em Proof: \ }}
\newcommand{\qed}{\hspace*{5 mm}$\square$\bigskip}
\newenvironment{proof}{\noindent {\pr}\ }{\qed}

\newcommand{\codi}{{\cal C}}
\newcommand{\pes}{\mbox{wt}}
\newcommand{\pesLee}{\mbox{w$_L$}}

\newcommand{\Z}{{\mathbb{Z}}}
\newcommand{\N}{{\mathbb{N}}}
\newcommand{\F}{\mathbb{F}}
\newcommand{\Fn}{\F^{\:n}}

\newcommand{\zero}{{\mathbf 0}}
\renewcommand{\u}{{\mathbf 1}}
\newcommand{\dd}{\displaystyle}

\title{Propelinear structure of $\Z_{2k}$-linear codes}

\author{J. Borges, C. Fern{\'a}ndez-C{\'o}rdoba and J. Rif{\`a} \\ 
{\small Department of Information and Communications Engineering,} \\
{\small Universitat Aut{\`o}noma de Barcelona, 08193-Bellaterra, Spain}\\
{\small (e-mail: \{jborges, cfernandez, jrifa\}@deic.uab.cat)}} 

\date{\today}

\begin{document}
\maketitle

\begin{abstract}
Let $\codi$ be an additive subgroup of $\Z_{2k}^n$ for any $k\geq 1$. We define a Gray map $\Phi:\Z_{2k}^n \longrightarrow \Z_2^{kn}$ 
such that $\Phi(\codi)$ is a binary propelinear code and, hence, a Hamming-compatible group code. Moreover, $\Phi$ is the 
unique Gray map such that $\Phi(C)$ is Hamming-compatible group code. Using this Gray map we discuss about the nonexistence of 1-perfect
binary mixed group code.
\end{abstract}

%%%%
\section{Introduction}
%%%%

Since the famous paper \cite{Sole} on $\Z_4$-linear codes, a large number of articles about $\Z_4$-linear and $\Z_k$-modulo codes have 
appeared. In more recent papers the Gray map, binary interpretation and concepts introduced in \cite{Sole} have been generalized, as in 
\cite{Carlet} for example.

In \cite{Jaume}, it is shown that linear and $\Z_4$-linear codes are subclasses of the more general class of translation invariant 
propelinear codes. In this paper we prove that any $\Z_{2k}$-modulo code is a binary propelinear code, but not translation invariant for 
$k>2$.

The paper is organized as follows. In Section 2 we give the preliminary concepts on distance compatibility, propelinear codes and 
translation invariant propelinear codes. In Section 3 we show the correspondence between $\Z_{2k}$-modulo codes and binary propelinear 
codes. In section 4 we define mixed group codes and we generalize the above correspondence for these codes to study which of them can 
be perfect. Finally, in Section 5 we point out some remarks and conclusions.

%%%%
\section{Propelinear codes}
%%%%

Let $\Fn$ be the $n$-dimensional binary vector space. We denote by $\zero$ the all-zero vector. As usual, the {\em (Hamming) distance} 
between two vectors $x,y\in\Fn$ is the number of coordinates in which they differ and denoted by $d(x,y)$. The {\em weight} of a vector 
$x\in\Fn$ is the number of its nonzero entries $\pes(x)=d(\zero,x)$.

The concept of {\em (Hamming) distance-compatible operation} in $\Fn$ is defined in \cite{Rifa} and \cite{Adds}. If 
$\;*:\Fn\times\Fn \longrightarrow \Fn$ is such an operation, then for all $v\in\Fn$ it should verify:
\begin{itemize}
\item[(i)] $d(v,v*e)=1\;\;\;\forall e\in\Fn \mbox{ with wt}(e)=1$;
\item[(ii)] $v*\zero = \zero*v=v$;
\item[(iii)] $v*e=w*e$ if and only if $v=w$, for all
$e\in\Fn \mbox{ with wt}(e)=1$.
\end{itemize}

If $(\Fn,*)$ is a group, then the operation $*$ is distance-compatible if and only if $d(v,v*u)=\pes(u)$ for all vectors 
$u,v\in\Fn$. The `if' part is trivial and the `only if' part is shown in \cite[Proposition 14]{Adds}. A binary code $\codi$ of length 
$n$ is a subset of $\F^n$. If this subset is a linear subspace of $\F^n$, then $\codi$ will be a linear code. In any case we will call 
the vectors in $\codi$ codewords. We denote by $(\codi,\star)$ a code in $\Fn$ with a group structure defined by $\star$. 
This operation could be nondefined in the whole space $\Fn$, but it could induce an action $\;\star:\codi\times\Fn \longrightarrow\Fn$. 

\begin{defi}
Let $(\codi,\star)$ a code in $\Fn$ and assume the operation $\star$ induces an action $\;\star:\codi\times\Fn \longrightarrow\Fn$. 
The action $\star$ is {\em Hamming-compatible}  if $d(x,x\star v)= ${\em $\pes$}$ (v)$, for all $x\in\codi$ and for all $v\in\Fn$.
\end{defi}

\begin{defi}
A binary code $(\codi,\star)$ of length $n$ is a {\em Hamming-compatible group code} if $(\codi,\star)$ is a group and it is possible to 
extend $\star:\codi\times\Fn\longrightarrow\Fn$ to a Hamming-compatible action.
\end{defi}

Of course, given a code $\codi\subset\Fn$ among all the different group structures, we are interested in those being Hamming-compatible 
(assuming we are working with the Hamming metric). A very general class of such codes are the propelinear ones, defined in \cite{JosHu}:

\begin{defi}
Let ${\cal S}_n$ be the symmetric group of permutations on $n$ elements. A (binary) code ${\cal C}$ of length $n$ is said to be 
{\em propelinear} if for any codeword $x\in {\cal C}$ there is a coordinate permutation $\pi_x\in {\cal S}_n$ verifying the properties:
\begin{enumerate}
\item $x+\pi_x(y)\in {\cal C}\;\;$ if $y\in {\cal C}$.
\item $\pi_x\circ \pi_y=\pi_z\;\;\;\forall y\in {\cal C}$, where $z=x+\pi_x(y)$.
\end{enumerate}
\end{defi}

Now, we can define the binary operation $\;\star:{\cal C}\times \Fn\longrightarrow \Fn$ such that
$$
x\star y=x+\pi_x(y)\;\;\forall x\in {\cal C}\;\;\forall y\in \Fn.
$$
This operation is clearly associative and closed in ${\cal C}$. Since, for any codeword $x\in {\cal C}$, $x\star y=x\star z$ implies $y=z$, 
we have that $x\star y\in {\cal C}$ if and only if $y\in {\cal C}$. Thus, there must be a codeword $e$ such that $x\star e=x$. It follows 
that $e=\zero$ is a codeword and, from 2, we deduce that $\pi_{\zero}$ is the identity permutation. Hence, $({\cal C},\star)$ is a group, 
which is not Abelian in general; $\zero$ is the identity element in ${\cal C}$ and $x^{-1}=\pi_x^{-1}(x)$, for all $x\in {\cal C}$. Note 
that $\Pi=\{\pi_x\mid x\in {\cal C}\}$ is a subgroup of ${\cal S}_n$ with the usual composition of permutations.

\begin{lemm}\label{Lem1}
Let $(\codi,\star)$ be a propelinear code, then
$$
d(x\star u,x\star v)=d(u,v)\;\;\;\forall x\in\codi\;\;\forall u,v\in\Fn
$$
\end{lemm}

\begin{proof}
The claim is trivial and can be found in \cite{JosHu} or \cite{Adds}.
\end{proof}

\begin{lemm}
A binary propelinear code is a Hamming-compatible group code.
\end{lemm}

\begin{proof}
Let $(\codi,\star)$ be such a code. We only have to prove that the action $\star:\codi\times\Fn\longrightarrow\Fn$ is Hamming-compatible. 
But this is clearly true because for any $x\in\codi$ and any $v\in\Fn$ we have
$$
d(x,x\star v)=d(x\star\zero,x\star v)=d(\zero,v)=\pes(v)
$$
applying Lemma \ref{Lem1}.
\end{proof}

A propelinear code $({\cal C},\star)$ is said to be a {\em translation invariant code} \cite{Jaume} if
$$
d(x,y)=d(x\star u,y\star u)\;\;\;\forall x,y\in {\cal C}
\;\;\;\forall u\in\Fn.
$$

As can be seen in \cite{Jaume} the class of translation invariant propelinear codes includes linear and $\Z_4$-linear codes. 
In fact, any translation invariant propelinear code of length $n$ can be viewed as a group isomorphic to a subgroup of 
${\mathbb Z}_2^{k_1} \oplus{\mathbb Z}_4^{k_2} \oplus {\cal Q}_8^{k_3}$; where $k_1+2k_2+4k_3=n$ and ${\cal Q}_8$ is the quaternion 
group on eight elements.
Clearly, the class of propelinear codes is more general than the class of linear codes, 
being in this case $\pi_x=Id$ for any codeword. We can find other exemples of propelinear structure,
 for instance, in \cite{Sole} there are exemples of $\Z_4$-linear codes (Goethals, Preparata like,...) and in 
\cite{Jaume} we can find the propelinear structure of the standard Preparata code which is not a $\Z_4$-linear code.

%%%%
\section{$\Z_{2k}$-codes as propelinear codes}
%%%%

There are different ways of giving a generalization of a Gray may. For instance, Carlet gives in \cite{Carlet} a generalization to 
$\Z_{2^k}$. In this paper we will give one preserving the basic property that the distance beetwen the images of two consecutive elements is exactly 
one (see \cite{Sole}).

\begin{defi}
The {\em Lee weight} of an element $x\in \Z_k$, {\em $\pesLee(x)$}, is defined as the minimum absolute value of any representative 
of its class in $\Z_k$. The {\em Lee distance} between $x,y\in \Z_k$ is $d_L(x,y)=${\em $\pesLee$}$(x-y)$. Clearly,
{\em $\pesLee$}$(x)=d_L(x,0)$.
\end{defi}

\begin{defi}
A Gray map is an aplication $\varphi:\Z_r\longrightarrow\Z_2^m$ such that
\begin{enumerate}
\item[(i)] $\varphi$ is one-to-one,
\item[(ii)] $d(\varphi(i),\varphi(i+1))=1$, $\forall i \in \Z_r$.
\end{enumerate}
\end{defi}

\begin{lemm}
Let  $\varphi:\Z_r\longrightarrow\Z_2^m$ a Gray map, then $r$ is even. 
\end{lemm}

\begin{proof}
Let $\psi:\Z_r\longrightarrow\Z_2$ defined as $\psi(i)=\pes(\varphi(i)) \bmod{2}$. Clearly, we can write $\psi(i)=\psi(0)+i \bmod{2}$. 
By definition of Gray map we have $d(\varphi(r-1),\varphi(0))=1$ but, if $r$ is odd, $\psi(r-1)=\psi(0)+r-1=\psi(0) \bmod{2}$ which is 
a contradiccion.
\end{proof}

\medskip

\begin{defi}
Let $\varphi:\Z_{2k}\longrightarrow\Z_2^m$ be a Gray map. $\varphi$ is {\rm distance-preserving} if $d(\varphi(i),\varphi(j))=d_L(i,j)$ 
and it is {\rm weight-preserving} if {\em $\pes$}$(\varphi(i))=${\em $\pesLee$}$(i)$.
\end{defi}

\begin{defi}
Let $\varphi:\Z_{2k}\longrightarrow\Z_2^m$ a Gray map and let + the usual operation in $\Z_{2k}$. We define the operation $\cdot$ in 
$\varphi(\Z_{2k})$ as:
\begin{equation}
\label{eq:product1}
\varphi(i)\cdot \varphi(j)=\varphi(i+j)
\end{equation}
for all $i,j\in \Z_{2k}$.
\end{defi}

\begin{lemm} \label{lemm:w_pres}
Let $\varphi:(\Z_{2k},+)\longrightarrow(\Z_2^m,\cdot)$ a Gray map such that $(\varphi(\Z_{2k}), \cdot)$ is a Hamming-compatible code. 
Then $\varphi$ is distance-preserving if and only if $\varphi$ is weight-preserving.
\end{lemm}

\begin{proof}
Clearly, if $\varphi$ is distance-preserving then is weight-preserving by definition of $\pes$ and $\pesLee$. \\
Suppose $\varphi$ is weight-preserving, then
\[
d(\varphi(i),\varphi(j))=d(\varphi(i),\varphi(i)\varphi(j-i))=\pes(\varphi(j-i))=w_L(j-i)=d_L(i,j)
\] 
\end{proof}

\begin{lemm}
Let $\varphi:\Z_{2k}\longrightarrow\Z_2^m$ be a Gray map such that $(\varphi(\Z_{2k}),\cdot)$ is a Hamming-compatible code, then 
$\varphi(0)=\zero$.
\end{lemm}
\begin{proof}
Let $a\in \Z_{2k}$. We know $(\varphi(\Z_{2k}),\cdot)$ is Hamming-compatible so,
\[
0=d(\varphi(a), \varphi(a+0))=d(\varphi(a),\varphi(a)\cdot \varphi(0))= \pes(\varphi(0))
\]
Now, by definition of $\pes()$, we have $\varphi(0)=\zero$.
\end{proof}

\begin{theo}\label{dist-pr}
Let $\varphi:\Z_{2k}\longrightarrow\Z_2^m$ be a Gray map. If $(\varphi(\Z_{2k}),\cdot)$ is a Hamming-compatible code, 
then $\varphi$ is distance-preserving.
\end{theo}

\begin{proof}
By the lemma (\ref{lemm:w_pres}), we only must proof that $\varphi$ is weight-preserving.\\
Clearly, $\varphi(0)=\zero$, $\varphi(1)=e_{i_1}$ and $\varphi(2)=e_{i_1}+e_{i_2}$ where $e_{i_s} \in \Z_2^m$ is the vector with $1$ in 
the coordinate $i_s$ and $0$ elsewhere. Let $j\in\Z_{2k}$ such that $\varphi(t)=e_{i_1}+\cdots+e_{i_t}$ $\forall t \leq j$ and 
$\pes(\varphi(j+1))=j-1$ ($j$ exists because $\varphi(2k-1)=1$).\medskip

If $j=k$ then $\pes(\varphi(i))=i=\pesLee(i)$ $\forall i\leq k$ and 
$\pes(\varphi(k+i))=d(\varphi(k),\varphi(2k+i))=d(\varphi(k),\varphi(i))=d(\varphi(k),\varphi(k)\cdot \varphi(k-i))=\pes(\varphi(k-i))=k-i=\pesLee(k+i) \; \forall i\leq k$. 
So, if $j=k$, the proof is finished.\medskip

Suppose $j<k$.
There exists $r\geq 1$ such that $\pes(\varphi(j+i))=\pes(\varphi(j+i-1))-1 \; \forall i\leq r$ and 
$\pes(\varphi(j+r+1))=\pes(\varphi(j+r))+1$. As we know, $d(\varphi(i),\varphi(j+i))=\pes(\varphi(j))=j$, therefore 
$\pes(\varphi(j+i))=j-i \; \forall i\leq r$. If $r=j$ then $\pes(\varphi(j+r))=0$ which is not possible because of the one-to-one
condition of the Gray map. Then $r<j$ and $\pes(\varphi(j+r))>1$.

In the same way, there exists $s \geq 1$ such that $\pes(\varphi(j+r+i))=\pes(\varphi(j+r+i-1))+1 \; \forall i\leq s$ and 
$\pes(\varphi(j+r+s+1)))=\pes(\varphi(j+r+s))-1$. As we know, $d(\varphi(i),\varphi(j+r+i))=\pes(\varphi(j+r))=j-r$, therefore 
$\pes(\varphi(j+r+i))=\pes(\varphi(j+r))+i= j-r+i$. If $s=r$ then $\varphi(j+r+s)=\varphi(j)$ which is not possible, so $s<r$.\\

We can use the same argument starting from $j+r+s$ and we always obtain images in $\Z_2^m$ with weights w such that $1>$w$>j$. This is a 
contradiction with the fact that $\pes(\varphi(2k-1))=1$.
\end{proof}

\bigskip

Let $\codi$ be a subgroup of $(\Z_{2k}^n,+)$ for some $k,n\geq 1$, where + is the usual addition in $\Z_{2k}$ coordinatewisely extended. 
We say that $\codi$ is a $\Z_{2k}$-modulo code or, briefly, a $\Z_{2k}$-code.
We will see a binary representation of any such code as a propelinear code.

Let $\zero^{(i)}$ be the all-zero vector of length $i$ and let $\u^{(j)}$
be the all-one vector of length $j$. We denote by `$\mid$' the concatenation,
i.e. if $x=(x_1,\ldots,x_r)$ and $y=(y_1,\ldots,y_s)$, then
$(x\mid y)=(x_1,\ldots,x_r,y_1,\ldots,y_s)$.\\
\bigskip
Define the Gray map $\phi:\Z_{2k}\longrightarrow\Z_2^k$ such that:
\begin{equation}
\begin{array}{cc}
\label{phi}
(i)& \phi(i)=(\zero^{(k-i)}\mid\u^{(i)})\;\;\;\forall i=0,\ldots,k-1, and \\
(ii)& \phi(i+k)=\phi(i)+\u^{(k)}\;\;\;\forall i=0,\ldots,k-1.
\end{array}
\end{equation}
Define also the associated permutation on $k$ coordinates
\begin{equation}
\label{sigma}
\sigma_j=(1,k,k-1,\ldots,2)^j
\end{equation}
(i.e. $j$ left shifts)
for all vector
$\phi(j)$, $j=0,\ldots,2k-1$.\\

Note that this Gray map $\phi$ is distance-preseving and weight-preserving.

\begin{defi}
 Let $\phi$ be the Gray map defined in (\ref{phi}). For any two elements $\phi(i),\phi(j)\in \phi(\Z_{2k})$ define the product
\begin{equation}
\label{eq:product2}
\phi(i)\cdot\phi(j)=\phi(i)+\sigma_i(\phi(j)) 
\end{equation}
\end{defi}

We are going to prove that the above product is, in fact, the one defined in (\ref{eq:product1}).

\begin{lemm} Let $\phi$ be the Gray map defined in (\ref{phi}). Let $\phi(i)\in \phi(\Z_{2k})$ and $\cdot$ the product defined in 
(\ref{eq:product2}). Then
$$
\phi(i)=\phi(1)^i
%(\cdots((\phi(1)\phi(1))\phi(1))\cdots\phi(1))=(\phi(1)(\cdots(\phi(1)(\phi(1)\phi(1)))\cdots))=\phi(1)^i.
$$
\end{lemm}

\begin{proof}
It is easy to verify that $\phi(i)=\phi(i-1)\cdot\phi(1)=\phi(1)\cdot\phi(i-1)$. Appliying
this repeadly yields the result.
\end{proof}

\begin{prop}
$(\phi(\Z_{2k}),\cdot)$ is a group, with $\phi$ and $\cdot$ defined in (\ref{phi}) and (\ref{eq:product2}) respectively.
\end{prop}

\begin{proof}
We have that
$$
(\phi(i)\cdot\phi(j))\cdot\phi(\ell)=(\phi(1)^i\cdot\phi(1)^j)\cdot\phi(1)^\ell=\phi(1)^{i+j+\ell}
=\phi(i)\cdot(\phi(j)\cdot\phi(\ell)),
$$
for all $i,j,\ell\in\Z_{2k}$. Therefore, the operation is associative.

It is clear that $\zero^{(k)}=\phi(0)$ acts as the identity element. On the other hand,
given $\phi(i)\in\phi(\Z_2^k)$, we have that
$$
\phi(i)\cdot\phi(k-i)=\phi(1)^{i+k-i}=\phi(1)^k=\phi(k)=\phi(0)=\zero^{(k)}.
$$
\end{proof}

\begin{coro}
 Let $\phi$ defined in (\ref{phi}) and $\cdot$ the operation given in (\ref{eq:product2}). The map 
$\;\phi:(\Z_{2k},+)\longrightarrow (\phi(\Z_{2k}),\cdot)$ is a group homomorphism and so the operation in (\ref{eq:product1}) and 
(\ref{eq:product2}) are the same.
\end{coro}

\begin{proof}
Given $i,j\in\Z_{2k}$, we have
$$
\phi(i+j)=\phi(1)^{i+j}=\phi(1)^i\cdot\phi(1)^j=\phi(i)\cdot\phi(j).
$$
\end{proof}

\begin{theo}
\label{unic}
Let $\varphi:\Z_{2k}\longrightarrow\Z_2^l$ be a Gray map. If $(\varphi(\Z_{2k}),\cdot)$ is a Hamming-compatible code where $\cdot$ is
the operation defined in (\ref{eq:product1}), then $\varphi$ is unique up to coordinate permutation.
\end{theo}

\begin{proof}
($\varphi(\Z_{2k}),\cdot)$ is a Hamming-compatible code and, by Theorem \ref{dist-pr}, $\varphi$ has the following properties:
\begin{itemize}
\item $\varphi(j)=e_{i_1}+\cdots +e_{i_j}$, for $j=1,\ldots,k$.
\item $\varphi(j+k)=\u^{(k)} +e_{i_1}+\cdots +e_{i_j}$, for $j=1,\ldots,k$.
\end{itemize}   
 where $e_{i_s} \in \Z_2^l$ is the vector with $1$ in the coordinate $i_s$ and $0$ elsewhere.\\
For $j=1,\ldots,k$, let $\mu_j$ be the transposition such that $\mu_j(e_{i_j})=e_{k-j+1}$. Let $\mu$ be the permutation whose 
decomposition in product of transpositions is $\mu_1\cdot\ldots\cdot\mu_k$. \\
Now it is easy to check that  $\phi=\mu\circ\varphi$, where $\phi$ is the map defined in (\ref{phi}).

\end{proof}

\begin{rema}
If $\varphi:\Z_{2k}\longrightarrow\Z_2^l$ is a Gray map, we have $l\geq k $ and, by the last theorem, if $l>k$ there are useless 
coordinates. Thus we can assume $l=k$.
\end{rema}

\begin{defi}\label{def:ext-Phi}
We define the extended map
$\Phi:\Z_{2k}^n\longrightarrow\Z_2^{kn}$ such that
$\Phi(j_1,\ldots,j_n)=(\phi(j_1),\ldots,\phi(j_n))$, where $\phi$ is defined in (\ref{phi}). Finally, we define
the permutations $\pi_x=(\sigma_{j_1}|\cdots|\sigma_{j_n}$), for $x=\Phi(j_1,\ldots,j_n)$, 
where $\sigma_i$ is defined in (\ref{sigma}). 
\end{defi}

Next theorem will prove that given a $\Z_{2k}$-code of length $n$, there exists a propelinear code of lenght $kn$ such that both codes
are isomorphic. The isomorphism beetwen them extends the usual structure in $\Z_{2k}$ (+) to the propelinear structure 
in $\Z_2^k$.

\begin{theo}\label{theo:isocodes}
If $\codi$ is a $\Z_{2k}$-code, then $\Phi(\codi)$ is a propelinear code with
associated permutation $\pi_x$ for all codeword $x\in\Phi(\codi)$.
\end{theo}

\begin{proof}
Let $x=\Phi(j_1,\ldots,j_n)=(\phi(j_1),\ldots,\phi(j_n))$ and
$y=\Phi(i_1,\ldots,i_n)=(\phi(i_1),\ldots,\phi(i_n))$ be two codewords.
Then,
$$
x+\pi_x(y)=(\phi(j_1)+\sigma_{j_1}(\phi(i_1)),\ldots,
\phi(j_n)+\sigma_{j_n}(\phi(i_n)).
$$
For any coordinate, say $r$, we have that
$$
\phi(j_r)+\sigma_{j_r}(\phi(i_r))=\phi(1)^{j_r}\phi(1)^{i_r}=\phi(1)^{
j_r+i_r}=\phi(j_r+i_r).
$$
Thus,
$$
x+\pi_x(y)=(\phi(j_1+i_1),\ldots,\phi(j_n+i_n))=\Phi((j_1,\ldots,j_n)
+(i_1,\ldots,i_n)).
$$
Therefore, it is clear that $x+\pi_x(y)\in\Phi(\codi)$.

On the other hand, the associated permutation of $\phi(j_r+i_r)$ is 
$$
\sigma_{j_r+i_r}=(1,k,k-1,\ldots,2)^{j_r+i_r}=\sigma_{j_r}\circ\sigma_{i_r},
$$
hence, if $z=x+\pi_x(y)$, then $\pi_z=\pi_x\circ\pi_y$.
\end{proof}

\begin{coro}
The map $\Phi:(\codi,+)\longrightarrow (\Phi(\codi),\star)$ is a group
isomorphism, where $x\star y=
x+\pi_x(y)$ for all $x,y\in\Phi(\codi)$.
\end{coro}

\begin{proof}
As we have seen in the previous proof, $x\star y=\Phi(\Phi^{-1}(x)+\Phi^{-1}(y))$ and, clearly, $\Phi$ is bijective.
\end{proof}

In \cite{Jaume} it is shown that linear and $\Z_4$-linear codes are
translation invariant. Now, we show that for $k>2$ any $\Z_{2k}$-code,
viewed as a binary propelinear code, is not translation invariant according
to the classification given in \cite{Jaume}.

\begin{prop}
If $k>2$ and $\codi\in\Z_{2k}^n$, then $\Phi(\codi)$ is a propelinear but
not translation invariant code.
\end{prop}

\begin{proof}
Consider the vector $z=(1,0,\ldots,0,1)\in\F^{\:k}$. Then it is easy to check
that $d(\zero^{(k)}\star z,\phi(1)\star z)=3\neq d(\zero^{(k)},\phi(1))=1$.
\end{proof}

\medskip

We have seen that starting from a $\Z_{2k}$-code $\codi$, of length $n$, the code $\Phi(\codi)$ with $\Phi$ (see Definition \ref{def:ext-Phi}) is a propelinear code of 
length $kn$ and both codes are isomorphic (Theorem \ref{theo:isocodes}). As we defined $\Phi$, the minimum Hamming distance in
$\Phi(\codi)$ is exactly the minimum Lee distance in $\codi$ but it is at least the minimum Hamming distance in $\codi$. 

Let $N$ be the number of codewords of $\codi$; clearly, it is also the codewords number of $\Phi(\codi)$.
\\
Let $R=\dd\frac{log_{2k}N}{n}=\dd\frac{log_2N}{n\cdot log_22k}=\dd\frac{log_2N}{n(1+log_2k)}$ 
be the information rate of $\codi$, and let $R'$ the information rate of $\Phi(\codi)$.
We can express $R'$ as
 \[
 R'=\dd\frac{log_2N}{kn}=\dd\frac{1+log_2k}{k}R
\]
therefore $R'$ is getting smaller than $R$ while the value of $k$ is raising; in fact, if $k\geq 3$ we obtain $R'<R$. 

\medskip

In this section we have seen that $\Z_{2k}$-codes can be represented as binary codes. We will use this representation in the next 
section to give some results about codes in $\Z_{2i_1}^{k_1}\times\cdots\times \Z_{2i_r}^{k_r}$, where $\times$ denotes the direct product,
and some necessary conditions to be 1-perfect codes.

%%%%%%%%%%%%%%%

%%%%
\section{Perfect propelinear codes}
%%%%

\begin{defi}
A general mixed group code $\codi$ is an additive subgroup of $G_1\times\cdots\times G_r$, where $G_1,\ldots,G_r$ are finite groups. 
We say that a binary code $\codi$ of length $n$ is a {\em mixed group code of type} $(\Z_{2i_1}^{k_1},\ldots,\Z_{2i_r}^{k_r})$ 
if $\codi=\Phi(C)$, where $i_1,\cdots,i_r$ are the minimum value such that $C$ is a subgroup of 
$\Z_{2i_1}^{k_1}\times\cdots\times\Z_{2i_r}^{k_r}$ and $\sum_{j=1}^r i_j k_j=n$. 
We denote $C\leq \Z_{2i_1}^{k_1}\times\cdots\times\Z_{2i_r}^{k_r}$.
\end{defi}

\begin{rema}
If $C \leq \Z_{2i_1}^{k_1}\times\cdots\times \Z_{2i_r}^{k_r}$ then $C=C_1 \times \cdots \times C_r$, with 
$C_j \leq \Z_{2i_j}^{k_j}$. We can write $\Phi(C)$ as $(\Phi_1(C_1),\cdots,\Phi_r(C_r))$ with  
$\Phi_j:\Z_{2i_j}^{k_j}\longrightarrow \Z_2^{k_j i_j}$ as in Definition \ref{def:ext-Phi}.
We will denote $x \in \codi$ as $(x_1|\cdots|x_r)$ where $x_j\in \Phi_j(C_j)$.
\end{rema}

\begin{theo} 
Let $\codi$ be a binary mixed group code of type $(\Z_{i_1}^{k_1},\ldots,\Z_{i_r}^{k_r})$ and length $n$. If $\codi$ is 1-perfect, then 
$\codi$ is of type $(\Z_2^k,\Z_4^{(n-k)/2})$ for some $k\in\N$. 
\end{theo} 

\begin{proof}
Let $\codi$ be a binary mixed group code of type $(\Z_{2i_1}^{k_1},\ldots,\Z_{2i_r}^{k_r})$. Suppose there exists $j \in \{1,\cdots,r\}$
such that $i_j>2$. Without loss of generality we will assume $j=1$ and $k_j=1$.

Let $x=(10\cdots 01|0\cdots0|\cdots |0\cdots0)\in \F^n$. If $\codi$ is 1-perfect, then there exists $y\in C$ such that
$d(x,\Phi(y))\leq 1$. As the minimum weight in $\codi$ is $3$ and the distance of $x$ must be at most $1$, the only possibility is 
$i_1=3$ and $\Phi(y)=(111|0\cdots0|\cdots |0\cdots0)$, therefore $C=G_1\times \cdots \times G_r$ where $G_1$ is a subgroup of $\Z_6$
and $3\in G_1$. The only subgroups of $\Z_6$ that contain $3$ are $\{0,3\}$ and $\Z_6$. We assume $G_1=\Z_6$; 
otherwise, $G_1=\{0,3\}$ would be isomorphic to $\Z_2$.
%Vector $(001|0\cdots0)$ is not a codeword because its weight is $1$ so, there exists two nonzero coordinates in $\{4,\cdots,n\}$.
%Without loss of generality we suppose these coordinates are exactly $4$ and $5$.
Let $u=(101100 \cdots 0)$, $v=(101010\cdots0)\in \F^n$ (where customary commas have been deleted); $u,v \not\in \codi$.
 The only codewords at distance $1$ of $u$ and $v$ are, respectively, $(111100\cdots 0)$ and $(111010\cdots 0)$ but the distance 
beetwen them is $2$ which is not possible if $\codi$ is 1-perfect.
\end{proof}

\medskip

1-perfect binary mixed codes of type $(\Z_2^k,\Z_4^{(n-k)/2})$ are called 1-perfect additive codes and they are studied in \cite{Adds}.

%%%%
\section{Conclusions}
%%%%

It is well known the usual Gray map from $\Z_4$ to $\Z_2^2$ (see \cite{Carlet}, \cite{Sole} and \cite{Wan}) but there are different ways of 
giving a generalization from $\Z_r$ to $\Z_2^m$. The generalization given in this paper has the property to be distance-preserving, 
considering the Lee distance in $\Z_r$ and the Hamming distance in $\Z_2^m$. However there could be other kind of generalizations, perhaps 
the most important to be considered are those where the distance in $\Z_r$ is different to the Lee distance or, merely, where the distance 
beetwen $0$ and $r-1$ is not $1$.

Let $\phi:\Z_r\longrightarrow \Z_2^m$ be the Gray map, and let $(\phi(\Z_r),\cdot)$ (defined in (\ref{eq:product1})) be a Hamming-compatible
code. We know that $r$ is even ($r=2k$) and, without useless coordinates, $m$ is exactly $k$. We have proved that such a Gray map is, 
in fact, unique up to coordinate permutation and we have used this to give some results on $\Z_{2k}$-codes.

Given a $\Z_{2k}$-code of length $n$, there exists a binary propelinear code of length $kn$ such that both codes are isomorphic. In this way
codes in $\Z_{2i_1}^{k_1}\times\cdots\times \Z_{2i_r}^{k_r}$ (or mixed groups of type ($\Z_{2i_1}^{k_1},\cdots ,\Z_{2i_r}^{k_r}$))
could be represented as binary codes. Finally we have seen that if such a code is 1-perfect then, necessarily, it is a code of type 
$(\Z_2^{k_1},\Z_4^{k_2})$. 

As we have seen at the end of the Section $3$, the representation of a $\Z_{2k}$-code as a binary code is not efficient enough because 
the information rate wich is $R$ in the first code, become $\dd\frac{1+log_2k}{k}R$ in the second one, that is lower.
From this point of view, as we have seen that the representation of a $\Z_{2k}$-code is unique, we should look for other alternatives, 
apart from Gray maps, to represent a $\Z_{2k}$-code as a binary code.

\end{document}